\documentclass[11pt,reqno]{amsart}

\usepackage{amsmath, amsfonts, amssymb,  mathrsfs,  array, stmaryrd,  indentfirst, amsthm,  hyperref, comment}

\usepackage[margin=1.0in]{geometry}
\usepackage{setspace}
\theoremstyle{plain}
\newtheorem{thm}{Theorem}[section]
\newtheorem{lemma}{Lemma}[section]
\newtheorem{defn}{Definition}[section]

\newtheorem{remark}{Remark}[section]

\numberwithin{equation}{section}

\makeatletter
\def\@setcopyright{}
\def\serieslogo@{}
\makeatother

\begin{document}

\title[A note on FTAP]{A note on the Fundamental Theorem of Asset Pricing under model uncertainty}\thanks{This research is supported by the National Science Foundation under grant  DMS-0955463.} 

\author{Erhan Bayraktar}
\address[Erhan Bayraktar]{Department of Mathematics, University of Michigan, 530 Church Street, Ann Arbor, MI 48104, USA}
\email{erhan@umich.edu}
\author{Yuchong Zhang }
\address[Yuchong Zhang]{Department of Mathematics, University of Michigan, 530 Church Street, Ann Arbor, MI 48104, USA}
\email{yuchong@umich.edu}
\author{Zhou Zhou }
\address[Zhou Zhou]{Department of Mathematics, University of Michigan, 530 Church Street, Ann Arbor, MI 48104, USA}
\email{zhouzhou@umich.edu}

\begin{abstract}
We show that the results of \cite{BN13} on the Fundamental Theorem of Asset Pricing and the super-hedging theorem can be extended to the case in which the options available for static hedging (\emph{hedging options}) are quoted with bid-ask spreads. In this set-up, we need to work with the notion of \emph{robust no-arbitrage} which turns out to be equivalent to no-arbitrage under the additional assumption that hedging options with non-zero spread are \emph{non-redundant}. A key result is the closedness of the set of attainable claims, which requires a new proof in our setting.
\end{abstract}
\keywords{Model uncertainty, bid-ask prices for options, semi-static hedging, non-dominated collection of probability measures, Fundamental Theorem of Asset Pricing, super-hedging, robust no-arbitrage, non-redundant options}

\maketitle

\section{introduction}
We consider a discrete time financial market in which stocks are traded dynamically and options are available for static hedging. We assume that the dynamically traded asset is liquid and trading in them does not incur transaction costs, but that the options are less liquid and their prices are quoted with a bid-ask spread. (The more difficult problem with transaction costs on a dynamically traded asset is analyzed in \cite{2013arXiv1309.1420B} and \cite{DS13}.) As in \cite{BN13} we do not assume that there is a single model describing the asset price behavior but rather a collection of models described by the convex collection $\mathcal{P}$ of probability measures, which does not necessarily admit a dominating measure. One should think of  $\mathcal{P}$ as being obtained from calibration to the market data. We have a collection rather than a single model because generally we do not have point estimates but a confidence intervals for the parameters of our models. Our first goal is to obtain a criteria for deciding whether the collection of models represented by $\mathcal{P}$ is viable or not. Given that $\mathcal{P}$ is viable we would like to obtain the range of prices for other options written on the dynamically traded assets. The dual elements in these result are martingale measures that price the hedging options correctly (i.e. consistent with the quoted prices). As in classical transaction costs literature, we need to replace the no-arbitrage condition by the stronger \emph{robust no-arbitrage condition}, as we shall see in Section~\ref{sec:FTAPI}. In Section~\ref{sec:FTAPII} we will make the additional assumption that the hedging options with non-zero spread are \emph{non-redundant} (see Definition \ref{defn:non-redundancy}). We will see that under this assumption no-arbitrage and robust no-arbitrage are equivalent. Our main results are Theorems~\ref{thm:main} and \ref{cor}.

\section{Fundamental Theorem with Robust No Arbitrage}\label{sec:FTAPI}
Let $S_t=(S^1_t, \ldots, S^d_t)$ be the prices of $d$ traded stocks at time $t\in\{0, 1,\ldots, T\}$ and $\mathcal{H}$ be the set of all predictable $\mathbb{R}^d$-valued processes, which will serve as our trading strategies. Let $g=(g^1, \ldots, g^e)$ be the payoff of $e$ options that can be traded only at time zero with bid price $\underline{g}$ and ask price $\overline{g}$, with $\overline{g} \geq \underline{g}$ (the inequality holds component-wise).  We assume $S_t$ and $g$ are Borel measurable, and there are no transaction costs in the trading of stocks. 
\begin{defn}[No-arbitrage and robust no-arbitrage]\label{defn:NA} 
We say that condition NA($\mathcal{P}$) holds if for all $(H,h)\in\mathcal{H}\times\mathbb{R}^e$,
\[H\bullet S_T+h^+(g-\overline{g})-h^-(g-\underline{g})\geq 0 \quad \mathcal{P}-\text{quasi-surely}\,(\text{-q.s.})\footnote{A set is $\mathcal{P}$-polar if it is $P$-null for all $P\in\mathcal{P}$. A property is said to hold $\mathcal{P}$-q.s. if it holds outside a $\mathcal{P}$-polar set.}\]
implies
\[H\bullet S_T+h^+(g-\overline{g})-h^-(g-\underline{g})= 0 \quad \mathcal{P}\text{-q.s.,}\]
where $h^{\pm}$ are defined component-wise and are the usual positive/negative part of $h$.\footnote{When we multiply two vectors, we mean their inner product.} 

We say that condition $NA^r(\mathcal{P})$ holds if there exists $\underline{g}',\overline{g}'$ such that $[\underline{g}', \overline{g}']\subseteq ri[\underline{g}, \overline{g}]$ and NA($\mathcal{P}$) holds if $g$ has bid-ask prices $\underline{g}', \overline{g}'$.\footnote{``ri" stands for relative interior. $[\underline{g}', \overline{g}']\subseteq ri[\underline{g}, \overline{g}]$ means component-wise inclusion.}
\end{defn}
\begin{defn}[Super-hedging price]
For a given a random variable $f$, its super-hedging price is defined as
\[\pi(f):=\inf\{x\in\mathbb{R}: \exists\,(H,h)\in\mathcal{H}\times\mathbb{R}^e \text{ such that } x+H\bullet S_T+h^+(g-\overline{g})-h^-(g-\underline{g})\geq f \ \mathcal{P}\text{-q.s.}\}.\]
Any pair $(H,h)\in\mathcal{H}\times\mathbb{R}^e$ in the above definition is called a semi-static hedging strategy. 
\end{defn}
\begin{remark}
\label{bounds_for_bid_ask_prices}
\textbf{[1]}
Let $\hat{\pi}(g^i)$ and $\hat{\pi}(-g^i)$ be the super-hedging prices of $g^i$ and $-g^i$, where the hedging is done using stocks and options excluding $g^i$. NA$^r(\mathcal{P})$ implies either
\[-\hat{\pi}(-g^i)\leq \underline{g}^i=\overline{g}^i\leq \hat{\pi}(g^i)\]
or
\begin{equation}\label{eq:NAr}
-\hat{\pi}(-g^i)\leq (\overline{g}')^i < \overline{g}^i \quad \text{and} \quad \underline{g}^i<(\underline{g}')^i \leq\hat{\pi}(g^i)
\end{equation}
where $\underline{g}', \overline{g}'$ are the more favorable bid-ask prices in the definition of robust no-arbitrage. The reason for working with robust no-arbitrage is to be able to have the strictly inequalities in \eqref{eq:NAr} for options with non-zero spread, which turns out to be crucial in the proof of the closedness of the set of hedgeable claims in \eqref{eq3} (hence the existence of an optimal hedging strategy), as well as in the construction of a dual element (see \eqref{eq1}).

\textbf{[2]} Clearly NA$^r(\mathcal{P})$ implies NA$(\mathcal{P})$, but the converse is not true. For example, assume in the market there is no stock, and there are only two options: $g_1(\omega)=g_2(\omega)=\omega,\ \omega\in\Omega:=[0,1]$. Let $\mathcal{P}$ be the set of probability measures on $\Omega$, $\underline g_1=\overline g_1=1/2$, $\underline g_2=1/4$ and $\overline g_2=1/2$. Then NA$(\mathcal{P})$ holds while NA$^r(\mathcal{P})$ fails.
\end{remark}

For $b,a\in\mathbb{R}^e$, let
\[
\mathcal{Q}^{[b,a]}:=\{Q\lll\mathcal{P}: Q \text{ is a martingale measure and } E^Q[g]\in [b,a]\}
\]
where $Q\lll\mathcal{P}$ means $\exists P\in\mathcal{P}$ such that $Q\ll P$.\footnote{$E^Q[g]\in [b,a]$ means $E^Q[g^i]\in[b^i, a^i]$ for all $i=1,\ldots,e$.} Let $\mathcal{Q}^{[b,a]}_\varphi:=\{Q\in\mathcal{Q}: E^Q[\varphi]<\infty\}$.
When $[b,a]=[\underline{g},\overline{g}]$, we drop the superscript and simply write $\mathcal{Q} ,\mathcal{Q}_\varphi$. Also define
\[\mathcal{Q}^{s}:=\{Q\lll\mathcal{P}: Q \text{ is a martingale measure and } E^Q[g]\in ri[\underline{g},\overline{g}]\}\]
and $\mathcal{Q}^{s}_\varphi:=\{Q\in\mathcal{Q}^s: E^Q[\varphi]<\infty\}$.

\begin{thm}\label{thm:main}
Let $\varphi\geq 1$ be a random variable such that $|g^i|\leq \varphi$ $\forall i=1,\ldots,e$. The following statements hold:
\begin{itemize}
\item[(a)] (Fundamental Theorem of Asset Pricing): The following statements are equivalent
\begin{itemize}
\item[(i)] NA$^r(\mathcal{P})$ holds.
\item[(ii)] There exists $[\underline{g}', \overline{g}']\subseteq ri[\underline{g}, \overline{g}]$ such that $\forall P\in\mathcal{P}$, $\exists Q\in\mathcal{Q}^{[\underline{g}', \overline{g}']}_\varphi$ such that $P\ll Q$.
\end{itemize}
\item[(b)] (Super-hedging) Suppose NA$^r(\mathcal{P})$ holds. Let $f:\Omega\rightarrow \mathbb{R}$ be Borel measurable such that $|f|\leq\varphi$. The super-hedging price is given by
\begin{equation}\label{eq}
\pi(f)=\sup_{Q\in\mathcal{Q}^s_\varphi}E^Q[f]=\sup_{Q\in\mathcal{Q}_\varphi}E^Q[f]\in(-\infty,\infty],
\end{equation}
and there exists $(H,h)\in\mathcal{H}\times\mathbb{R}^e$ such that $\pi(f)+H\bullet S_T+h^+(g-\overline{g})-h^-(g-\underline{g})\geq f \ \mathcal{P}\text{-q.s.}$.
\end{itemize}
\end{thm}

\begin{proof}
It is easy to show $(ii)$ in (a) implies that NA($\mathcal{P}$) holds for the market with bid-ask prices $\underline{g}', \overline{g}'$, Hence NA$^r(\mathcal{P})$ holds for the original market. The rest of our proof consists two parts as follows.\\
\textbf{Part 1: $\pi(f)>-\infty$ and the existence of an optimal hedging strategy in (b).} Once we show that the set
\begin{equation}\label{eq3}
\mathcal{C}_g:=\{H\bullet S_T+h^+(g-\overline{g})-h^-(g-\underline{g}):\ (H,h)\in\mathcal{H}\times\mathbb{R}^e\}-\mathcal{L}_+^0
\end{equation}
is $\mathcal{P}-q.s.$ closed (i.e., if $(W^n)_{n=1}^\infty\subset\mathcal{C}_g$ and $W^n\rightarrow W\ \mathcal{P}-q.s.$, then $W\in\mathcal{C}_g$), the argument used in the proof of \cite[Theorem 2.3]{BN13} would conclude the results in part 1. We will demonstrate the closedness of $\mathcal{C}_g$ in the rest of this part. 

Write $g=(u,v)$, where $u=(g^1,\dotso,g^r)$ consists of the hedging options without bid-ask spread, i.e, $\underline g^i=\overline g^i$ for $i=1,\dotso,r$, and $v=(g^{r+1},\dotso,g^e)$ consists of those with spread, i.e., $\underline g^i<\overline g^i$ for $i=r+1,\dotso,e$, for some $r\in\{0,\dotso,e\}$. Denote $\underline u:=(\underline g^1,\dotso,\underline g^r)$ and similarly for $\underline v$ and $\overline v$. Define 
$$\mathcal{C}:=\{H\bullet S_T+\alpha(u-\underline u):\ (H,\alpha)\in\mathcal{H}\times\mathbb{R}^r\}-\mathcal{L}_+^0.$$
Then $\mathcal{C}$ is $\mathcal{P}-q.s.$ closed by \cite[Theorem 2.2]{BN13}.

Let $W^n\rightarrow W\ \mathcal{P}-q.s.$ with
\begin{equation}\label{ee}
W^n=H^n\bullet S_T+\alpha^n(u-\underline{u})+(\beta^n)^+(v-\overline{v})-(\beta^n)^-(v-\underline{v})-U^n\in\mathcal{C}_g,
\end{equation}
where $(H^n,\alpha^n,\beta^n)\in\mathcal{H}\times\mathbb{R}^r\times\mathbb{R}^{e-r}$ and $U^n\in\mathcal{L}_+^0$. If $(\beta^n)_n$ is not bounded, then by passing to subsequence if necessary, we may assume that $0<||\beta^n||\rightarrow\infty$ and  rewrite \eqref{ee} as
$$\frac{H^n}{\beta^n}\bullet S_T+\frac{\alpha^n}{||\beta^n||}(u-\underline u)\geq\frac{W^n}{||\beta^n||}-\left(\frac{\beta^n}{||\beta^n||}\right)^+(v-\overline{v})+\left(\frac{\beta^n}{||\beta^n||}\right)^-(v-\underline{v})\in\mathcal{C},$$
where $||\cdot||$ represents the sup-norm. Since $\mathcal{C}$ is $\mathcal{P}-q.s.$ closed, the limit of the right hand side above is also in $\mathcal{C}$, i.e., there exists some $(H,\alpha)\in\mathcal{H}\times\mathbb{R}^r$, such that
$$H\bullet S_T+\alpha(u-\underline{u})\geq-\beta^+(v-\overline{v})+\beta^-(v-\underline{v}),\quad\mathcal{P}-a.s.,$$
where $\beta$ is the limit of $(\beta^n)_n$ along some subsequence with $||\beta||=1$. NA$(\mathcal{P})$ implies that
\begin{equation}\label{eq4}
H\bullet S_T+\alpha(u-\underline{u})+\beta^+(v-\overline{v})-\beta^-(v-\underline{v})=0,\quad\mathcal{P}-a.s..
\end{equation}
As $\beta=:(\beta_{r+1},\dotso,\beta_e)\neq 0$, we assume without loss of generality (w.l.o.g.) that $\beta_e\neq 0$. If $\beta_e<0$, then we have from \eqref{eq4} that
$$\underline g^e+\frac{H}{\beta_e^-}\bullet S_T+\frac{\alpha}{\beta_e^-}(u-\underline{u})+\sum_{i=r+1}^{e-1}\left[\frac{\beta_i^+}{\beta_e^-}(g^i-\overline{g}^i)-\frac{\beta_i^-}{\beta_e^-}(g^i-\underline{g}^i)\right]=g^e,\quad\mathcal{P}-a.s..$$
Therefore $\hat\pi(g^e)\leq\underline g_e$, which contradicts the robust no-arbitrage property (see \eqref{eq:NAr}) of $g^e$. Here $\hat\pi(g^e)$ is the super-hedging price of $g^e$ using $S$ and $g$ excluding $g^e$. Similarly we get a contradiction if $\beta_e>0$.

Thus $(\beta^n)_n$ is bounded, and has a limit $\beta\in\mathbb{R}^{e-r}$ along some subsequence $(n_k)_k$. Since by \eqref{ee}
$$H^n\bullet S_T+\alpha^n(u-\underline u)\geq W^n-(\beta^n)^+(v-\overline{v})+(\beta^n)^-(v-\underline{v})\in\mathcal{C},$$
the limit of the right hand side above along $(n_k)_k$, $W-\beta^+(v-\overline{v})+\beta^-(v-\underline{v})$, is also in $\mathcal{C}$ by its closedness, which implies $W\in\mathcal{C}_g$. \\
\textbf{Part 2: $(i)\Rightarrow (ii)$ in part (a) and \eqref{eq} in part (b).}  We will prove the results by an induction on the number of hedging options, as in \cite[Theorem 5.1]{BN13}. Suppose the results hold for the market with options $g^1, \ldots, g^e$. We now introduce an additional option $f\equiv g^{e+1}$ with $|f|\leq \varphi$, available at bid-ask prices $\underline{f}<\overline{f}$ at time zero. (When the bid and ask prices are the same for $f$, then the proof is identical to \cite{BN13}.) 

$(i)\implies (ii)$ in (a): Let $\pi(f)$ be the super-hedging price when stocks and $g^1, \ldots, g^e$ are available for trading. By NA$^r(\mathcal{P})$ and \eqref{eq} in part (b) of the induction hypothesis, we have
\begin{equation}\label{eq1}
\overline{f}>\overline{f}'\geq-\pi(-f)=\inf_{Q\in\mathcal{Q}^s_\varphi}E^Q[f] \quad\text{and}\quad \underline{f}<\underline{f}' \leq\pi(f)=\sup_{Q\in\mathcal{Q}^s_\varphi}E^Q[f]
\end{equation}
where $[\underline{f}',\overline{f}']\subseteq(\underline{f},\overline{f})$ comes from the definition of robust no-arbitrage. This implies that there exists $Q_+, Q_-\in\mathcal{Q}^s_\varphi$ such that $E^{Q_+}[f]>\underline{f}''$ and $E^{Q_-}[f]<\overline{f}''$ where $\underline{f}''=\frac{1}{2}(\underline{f}+\underline{f}')$, $\overline{f}''=\frac{1}{2}(\overline{f}+\overline{f}')$. By (a) of induction hypothesis, there exists $[b,a]\subseteq ri[\underline{g},\overline{g}]$ such that for any $P\in\mathcal{P}$, we can find $Q_0\in\mathcal{Q}^{[b,a]}_\varphi$ satisfying $P\ll Q_0\lll\mathcal{P}$. Define
\[\underline{g}'=\min(b,E^{Q_+}[g],E^{Q_-}[g]), \quad\text{and}\quad \overline{g}'=\max(a,E^{Q_+}[g],E^{Q_-}[g])\]
where the minimum and maximum are taken component-wise. We have $[b,a]\subseteq [\underline{g}',\overline{g}']\subseteq ri[\underline{g},\overline{g}]$ and $Q_+,Q_-\in \mathcal{Q}^{[\underline{g}',\overline{g}']}_\varphi$. 

Now, let $P\in\mathcal{P}$. (a) of induction hypothesis implies the existence of a $Q_0\in\mathcal{Q}^{[b,a]}_\varphi\subseteq\mathcal{Q}^{[\underline{g}',\overline{g}']}_\varphi$ satisfying $P\ll Q_0\lll\mathcal{P}$. Define 
\[Q:=\lambda_-Q_-+\lambda_0Q_0+\lambda_+Q_+.\]
Then $Q\in\mathcal{Q}^{[\underline{g}',\overline{g}']}_\varphi$ and $P\ll Q$. By choosing suitable weights $\lambda_-, \lambda_0,\lambda_+\in(0,1), \lambda_-+ \lambda_0+\lambda_+=1$, we can make $E^Q[f]\in [\underline{f}'',\overline{f}'']\subseteq ri[\underline{f},\overline{f}]$.

\eqref{eq} in (b): Let $\xi$ be a Borel measurable function such that $|\xi|\leq\varphi$. Write $\pi'(\xi)$ for its super-hedging price when stocks and $g^1, \ldots, g^e, f\equiv g^{e+1}$ are traded, $\mathcal{Q}'_\varphi:=\{Q\in \mathcal{Q}_\varphi: E^{Q}[f]\in [\underline{f},\overline{f}]\}$ and $\mathcal{Q}'^{s}_\varphi:=\{Q\in \mathcal{Q}^s_\varphi: E^{Q}[f]\in (\underline{f},\overline{f})\}$. We want to show 
\begin{equation}\label{goal}
\pi'(\xi)=\sup_{Q\in\mathcal{Q}'^{s}_\varphi} E^{Q}[\xi]=\sup_{Q\in\mathcal{Q}'_\varphi} E^{Q}[\xi].
\end{equation}
It is easy to see that
\begin{equation}\label{trivial_side}
\pi'(\xi)\geq\sup_{Q\in\mathcal{Q}'_\varphi} E^{Q}[\xi]\geq \sup_{Q\in\mathcal{Q}'^{s}_\varphi} E^{Q}[\xi]
\end{equation}
and we shall focus on the reverse inequalities. Let us assume first that $\xi$ is bounded from above, and thus $\pi'(\xi)<\infty$. By a translation we may assume $\pi'(\xi)=0$.

First, we show $\pi'(\xi)\leq \sup_{Q\in\mathcal{Q}'_\varphi} E^{Q}[\xi]$. 
It suffices to show the existence of a sequence $\{Q_n\}\subseteq\mathcal{Q}_\varphi$ such that $\lim_nE^{Q_n}[f]\in [\underline{f},\overline{f}]$ and $\lim_nE^{Q_n}[\xi]=\pi'(\xi)=0$. (See page 30 of \cite{BN13} for why this is sufficient.) 
In other words, we want to show that
\begin{equation}\notag
\overline{\{E^Q[(f,\xi)]: Q\in\mathcal{Q}_\varphi\}}\cap \left([\underline{f},\overline{f}]\times\{0\} \right)\neq\emptyset.
\end{equation}
Suppose the above intersection is empty. Then there exists a vector $(y,z)\in\mathbb{R}^2$ with $|(y,z)|=1$ that strictly separates the two closed, convex sets, i.e. there exists $b\in\mathbb{R}$ s.t.
\begin{equation}\label{eq:seperation}
\sup_{Q\in\mathcal{Q}_\varphi}E^Q[yf+z\xi]<b \ \text{ and }\ \inf_{a\in[\underline{f},\overline{f}]}ya>b.
\end{equation}
It follows that
\begin{equation}\label{1}
y^+\underline{f}-y^-\overline{f}+\pi'(z\xi)\leq\pi'(yf+z\xi)\leq \pi(yf+z\xi)=\sup_{Q\in\mathcal{Q}_\varphi}E^Q[yf+z\xi]< b< y^+\underline{f}-y^-\overline{f},
\end{equation}
where the first inequality is because one can super-replicate $z\xi=(yf+z\xi)+(-yf)$ from initial capital $\pi'(yf+z\xi)-y^+\underline{f}+y^-\overline{f}$, the second inequality is due to the fact that having more options to hedge reduces hedging cost, and the middle equality is by (b) of induction hypothesis. The last two inequalities are due to \eqref{eq:seperation}.

It follows from \eqref{1} that $\pi'(z\xi)<0$. Therefore, we must have that $z<0$, otherwise $\pi'(z\xi)=z\pi'(\xi)=0$ (since the super-hedging price is positively homogenous). Recall that we have proved in part (a) that $\mathcal{Q}'_\varphi\neq\emptyset$. Let $Q'\in\mathcal{Q}'_\varphi\subseteq\mathcal{Q}_\varphi$. 
The part of \eqref{1} after the equality implies that $yE^{Q'}[f]+zE^{Q'}[\xi]<y^+\underline{f}-y^-\overline{f}$. Since $E^{Q'}[f]\in[\underline{f},\overline{f}]$, we get $zE^{Q'}[\xi]<y^+(\underline{f}-E^{Q'}[f])-y^-(\overline{f}-E^{Q'}[f])\leq 0$. Since $z<0$, $E^{Q'}[\xi]>0$. But by \eqref{trivial_side}, $E^{Q'}[\xi]\leq \pi'(\xi)=0$, which is a contradiction. 

Next, we show $\sup_{Q\in\mathcal{Q}'_\varphi} E^{Q}[\xi]\leq \sup_{Q\in\mathcal{Q}'^{s}_\varphi} E^{Q}[\xi]$. It suffices to show for any $\varepsilon>0$ and every $Q\in\mathcal{Q}'_\varphi$, we can find $Q^s\in \mathcal{Q}'^{s}_\varphi$ such that $E^{Q^s}[\xi]>E^Q[\xi]-\varepsilon$. To this end, let $Q'\in\mathcal{Q}'^{s}_\varphi$ which is nonempty by part (a).
Define
\[Q^s:=(1-\lambda)Q+\lambda Q'.\]
We have $Q^s\lll\mathcal{P}$ by the convexity of $\mathcal{P}$, and $Q^s\in\mathcal{Q}'^{s}_\varphi$ if $\lambda\in(0,1]$. Moreover,
\[E^{Q^s}[\xi]=(1-\lambda)E^Q[\xi]+\lambda E^{Q'}[\xi]\rightarrow E^Q[\xi] \ \text{ as } \ \lambda\rightarrow 0.\]
So for $\lambda>0$ sufficiently close to zero, the $Q^s$ constructed above satisfies $E^{Q^s}[\xi]>E^Q[\xi]-\varepsilon$. Hence we have shown that the supremum over $\mathcal{Q}'_\varphi$ and $\mathcal{Q}'^{s}_\varphi$ are equal. This finishes the proof for upper bounded $\xi$.

Finally when $\xi$ is not bounded from above, we can apply the previous result to $\xi\wedge n$, and then let $n\rightarrow\infty$ and use the closedness of $\mathcal{C}_g$ in \eqref{eq3} to show that \eqref{eq} holds. The argument would be the same as the last paragraph in the proof of \cite[Thoerem 3.4]{BN13} and we omit it here.
\end{proof}

\section{A Sharper Fundamental Theorem with the non-redundancy assumption}\label{sec:FTAPII}

We now introduce the concept of non-redundancy. With this additional assumption we will sharpen our result.

\begin{defn}[Non-redundancy]\label{defn:non-redundancy}
A hedging option $g^i$ is said to be non-redundant if it is not perfectly replicable by stocks and other hedging options, i.e.
there does not exist $x\in\mathbb{R}$ and a semi-static hedging strategy $(H,h)\in\mathcal{H}\times\mathbb{R}^e$ such that
\[x+H\bullet S_T+\sum_{j\neq i}h^jg^j= g^i \ \mathcal{P}\text{-q.s.}.\]
\end{defn}
\begin{remark}
NA$^r(\mathcal{P})$ does not imply non-redundancy. For Instance, having only two identical options in the market whose payoffs are in $[c,d]$, with identical bid-ask prices $b$ and $a$ satisfying $b<c$ and $a>d$, would give a trivial counter example where NA$^r(\mathcal{P})$ holds yet we have redundancy. 
\end{remark}

\begin{lemma}\label{le1}
Suppose all hedging options with non-zero spread are non-redundant. Then NA$(\mathcal{P})$ implies NA$^r(\mathcal{P})$. 
\end{lemma}
\begin{proof}
Let $g=(g^1,\dotso,g^{r+s})$, where $u:=(g^1,\dotso,g^r)$ consists of the hedging options without bid-ask spread, i.e, $\underline g^i=\overline g^i$ for $i=1,\dotso,r$, and $(g^{r+1},\dotso,g^{r+s})$ consists of those with bid-ask spread, i.e., $\underline g^i<\overline g^i$ for $i=r+1,\dotso,r+s$. We shall prove the result by induction on $s$. Obviously the result holds when $s=0$. Suppose the result holds for $s=k\geq 0$. Then  for $s=k+1$, denote $v:=(g^{r+1}\dotso,g^{r+k})$, $\underline v:=(\underline g^{r+1},\dotso,\underline g^{r+k})$ and $\overline v:=(\overline g^{r+1},\dotso,\overline g^{r+k})$. Denote $f:=g^{r+k+1}$.
%

By the induction hypothesis, there exists $[\underline v',\overline v']\subset (\underline v,\overline v)$ be such that NA$(\mathcal{P})$ holds in the market with stocks, options $u$ and options $v$ with any bid-ask prices $b$ and $a$ satisfying $[\underline v',\overline v']\subset[b,a]\subset (\underline v,\overline v)$. Let $\underline v_n\in(\underline v,\underline v')$, $\overline v_n\in(\overline v',\overline v)$, $\underline f_n>\underline f$ and $\underline f_n<\overline f$, such that $\underline v_n\searrow\underline v$, $\overline v_n\nearrow\overline v$, $\underline f_n\searrow\underline f$ and $\overline f_n\nearrow \overline f$. We shall show that for some $n$, NA$(\mathcal{P})$ holds with stocks, options $u$, options $v$ with bid-ask prices $\underline v_n$ and $\overline v_n$, option $f$ with bid-ask prices $\underline f_n$ and $\overline f_n$. We will show it by contradiction.

If not, then for each $n$, there exists $(H^n,h_u^{n},h_v^{n},h_f^{n})\in\mathcal{H}\times\mathbb{R}^r\times\mathbb{R}^k\times\mathbb{R}$ such that
\begin{equation}\label{21}
H^n\bullet S_T+h_u^{n}(u-\underline u)+(h_v^{n})^+(v-\overline v_n)-(h_v^{n})^-(v-\underline v_n)+(h_f^{n})^+(f-\overline f_n)-(h_f^{n})^-(f-\underline f_n)\geq 0,\ \mathcal{P}-q.s.,
\end{equation}
and the strict inequality for the above holds with positive probability under some $P_n\in\mathcal{P}$. Hence $h_f^{n}\neq 0$. By a normalization, we can assume that $|h_f^{n}|=1$. By extracting a subsequence, we can w.l.o.g. assume that $h_f^{n}=-1$ (the argument when assuming $h_f^{n}=1$ is similar). If $(h_u^{n},h_v^{n})_n$ is not bounded, then w.l.o.g. we assume that $0<c^n:=||(h_u^{n},h_v^{n})||\rightarrow\infty$. By \eqref{21} we have that
$$\frac{H^n}{c^n}\bullet S_T+\frac{h_u^{n}}{c^n}(u-\underline u)+\frac{(h_v^{n})^+}{c^n}(v-\overline v_n)-\frac{(h_v^{n})^-}{c^n}(v-\underline v_n)-\frac{1}{c^n}(f-\underline f_n)\geq 0,\ \mathcal{P}-q.s..$$
By \cite[Theorem 2.2]{BN13}, there exists $H\in\mathcal{H}$, such that
$$H\bullet S_T+h_u(u-\underline u)+h_v^+(v-\overline v)-h_v^-(v-\underline v)\geq 0,\ \mathcal{P}-q.s.,$$
where $(h_u,h_v)$ is the limit of $(h_u^{n}/c^n,h_u^{n}/c^n)$ along some subsequence with $||(h_u,h_v)||=1$. NA$(\mathcal{P})$ implies that 
\begin{equation}\label{22}H\bullet S_T+h_u(u-\underline u)+h_v^+(v-\overline v)-h_v^-(v-\underline v)=0,\ \mathcal{P}-q.s..
\end{equation}
Since $(h_u,h_v)\neq 0$, \eqref{22} contradicts the non-redundancy assumption of $(u,v)$.

Therefore, $(h_u^{n},h_v^{n})_n$ is bounded, and w.l.o.g. assume it has the limit $(\hat h_u,\hat h_v)$. Then applying \cite[Theorem 2.2]{BN13} in \eqref{21}, there exists $\hat H\in\mathcal{H}$ such that
$$\hat H\bullet S_T+\hat h_u(u-\underline u)+\hat h_v^+(v-\overline v)-\hat h_v^-(v-\underline v)-(f-\underline f)\geq 0,\ \mathcal{P}-q.s..$$
NA$(\mathcal{P})$ implies that
$$\hat H\bullet S_T+\hat h_u(u-\underline u)+\hat h_v^+(v-\overline v)-\hat h_v^-(v-\underline v)-(f-\underline f)=0,\ \mathcal{P}-q.s.,$$
which contradicts the non-redundancy assumption of $f$.
\end{proof}

We have the following FTAP and super-hedging result in terms of NA$(\mathcal{P})$ instead of NA$^r(\mathcal{P})$, when we additionally assume the non-redundancy of $g$.

\begin{thm}\label{cor}
Suppose all hedging options with non-zero spread are non-redundant. Let $\varphi\geq 1$ be a random variable such that $|g^i|\leq \varphi$ $\forall i=1,\ldots,e$. The following statements hold:
\begin{itemize}
\item[(a')] (Fundamental Theorem of Asset Pricing): The following statements are equivalent
\begin{itemize}
\item[(i)] NA$(\mathcal{P})$ holds.
\item[(ii)] $\forall P\in\mathcal{P}$, $\exists Q\in\mathcal{Q}_\varphi$ such that $P\ll Q$.
\end{itemize}
\item[(b')] (Super-hedging) Suppose NA$(\mathcal{P})$ holds. Let $f:\Omega\rightarrow \mathbb{R}$ be Borel measurable such that $|f|\leq\varphi$. The super-hedging price is given by
\begin{equation}\label{eq}
\pi(f)=\sup_{Q\in\mathcal{Q}_\varphi}E^Q[f]\in(-\infty,\infty],
\end{equation}
and there exists $(H,h)\in\mathcal{H}\times\mathbb{R}^e$ such that $\pi(f)+H\bullet S_T+h^+(g-\overline{g})-h^-(g-\underline{g})\geq f \ \mathcal{P}\text{-q.s.}$.
\end{itemize}\end{thm}
\begin{proof}
(a')(ii)$\implies$(a')(i) is trivial. Now if (a')(i) holds, then by Lemma \ref{le1}, (a)(i) in Theorem \ref{thm:main} holds, which implies (a)(ii) holds, and thus (a')(ii) holds. Finally, (b') is implied by Lemma \ref{le1} and Theorem \ref{thm:main}(b).
\end{proof}

\begin{remark}Theorem~\ref{cor} generalizes the results of \cite{BN13} to the case when the option prices are quoted with bid-ask spreads. When $\mathcal{P}$ is the set of all probability measures and the given options are all call options written on the dynamically traded assets, a result with option bid-ask spreads similar to Theorem~\ref{cor}-(a) had been obtained by \cite{Cousot20073377}; see Proposition 4.1 therein, although the non-redundancy condition did not actually appear. (The objective of  \cite{Cousot20073377} was to obtain relationships between the option prices which are necessary and sufficient to rule out semi-static arbitrage and the proof relies on determining the correct set of relationships and then identifying a martingale measure.) 

However, the no arbitrage concept used in  \cite{Cousot20073377} is different: the author there assumes that there is no \emph{weak arbitrage} in the sense of \cite{MR2281789}; see also \cite{MAFI:MAFI12021} and \cite{Schachermayer13}.\footnote{The no-arbitrage assumption in \cite{Schachermayer13} is the model independent arbitrage of \cite{MR2281789}. However that paper rules out the model dependent arbitrage by assuming that a superlinearly growing option can be  bought for static hedging.} (Recall that a market is said to have weak arbitrage if for any given model (probability measure) there is an arbitrage strategy which is an arbitrage in the classical sense.) The arbitrage concept used here and in  \cite{BN13} is weaker, in that we say that a non-negative wealth ($\mathcal{P}$-q.s.) is an arbitrage even if there is a single $P$ under which the wealth process is a classical arbitrage. Hence our no-arbitrage condition is stronger than the one used in \cite{Cousot20073377}. But what we get out from a stronger assumption is the existence of a martingale measure $Q \in \mathcal{Q}_{\varphi}$ for each $P \in \mathcal{P}$. Whereas \cite{Cousot20073377} only guarantees the existence of only one martingale measure which prices the hedging options correctly.
\end{remark}

\bibliographystyle{plain}
\bibliography{FTAP_bib}{}

\begin{thebibliography}{1}

\bibitem{Schachermayer13}
B.~{Acciaio}, M.~{Beiglb{\"o}ck}, F.~{Penkner}, and W.~{Schachermayer}.
\newblock {A model-free version of the fundamental theorem of asset pricing and
  the super-replication theorem}.
\newblock {\em To appear in Mathematical Finance}.
\newblock Also available on ArXiv as 1301.5568.

\bibitem{2013arXiv1309.1420B}
E.~{Bayraktar} and Y.~{Zhang}.
\newblock {Fundamental Theorem of Asset Pricing under Transaction costs and
  Model uncertainty}.
\newblock {\em ArXiv e-prints}, September 2013.

\bibitem{BN13}
B.~{Bouchard} and M.~{Nutz}.
\newblock {Arbitrage and Duality in Nondominated Discrete-Time Models}.
\newblock {\em To appear in Annals of Applied Probability}.
\newblock Also available on ArXiv as 1305.6008.

\bibitem{Cousot20073377}
Laurent Cousot.
\newblock Conditions on option prices for absence of arbitrage and exact
  calibration.
\newblock {\em Journal of Banking \& Finance}, 31(11):3377 -- 3397, 2007.

\bibitem{MAFI:MAFI12021}
Mark Davis, Jan Obloj, and Vimal Raval.
\newblock Arbitrage bounds for prices of weighted variance swaps.
\newblock {\em To appear in Mathematical Finance}.

\bibitem{MR2281789}
Mark H.~A. Davis and David~G. Hobson.
\newblock The range of traded option prices.
\newblock {\em Math. Finance}, 17(1):1--14, 2007.

\bibitem{DS13}
Yan Dolinsky and H.~Mete Soner.
\newblock Robust hedging with proportional transaction costs.
\newblock {\em Finance Stoch.}, 18(2):327--347, 2014.

\bibitem{Tran.Cost}
Y.~Kabanov and M.~Safarian.
\newblock {\em Markets with Transaction Costs, Mathematical Theory}.
\newblock Springer-Verlag Berlin Heidelberg, 2009.

\end{thebibliography}

\end{document}